\documentclass[twocolumn]{IEEEtran}

\usepackage{hyperref}
\usepackage[width=.75\textwidth]{caption}
\usepackage{graphicx}
\usepackage{amsmath}
\usepackage{braket}
\usepackage{amsfonts}
\usepackage{amsthm}
\usepackage{dsfont}
\usepackage{float}
\usepackage{bbm}

\theoremstyle{definition}
\newtheorem{definition}{Definition}[section]

\newtheorem{theorem}{Theorem}[section]
\newtheorem{lemma}[theorem]{Lemma}

\newtheorem{conjecture}[theorem]{Conjecture}

\newcommand{\Z}{\mathbb{Z}}
\newcommand{\R}{\mathbb{R}}
\newcommand{\C}{\mathbb{C}}
\newcommand{\V}{\mathbb{Z}_2^n}
\newcommand{\supp}{\mathop{\mathrm{supp}}}
\newcommand{\nodes}{\mathop{\mathrm{nodes}}}
\usepackage[qm]{qcircuit}
\newcommand{\bigx}{\big}

\newcommand{\bigxxx}{\bigg}
\newcommand{\bigxxxx}{\Bigg}

\title{Non-Local to Local Eigenbasis Permutations of Pauli Product Diagonal Operators}
\author{Benjamin Commeau and Kevin Player}
\date{January 2024}

\begin{document}

\maketitle

\begin{abstract}
This paper investigates the feasibility of mapping non-local, sparse, diagonal forms of quantum Hamiltonians to local forms via eigenbasis permutations. We prove that such a mapping is not always possible, definitively refuting the "Quasiparticle Locality Conjecture \ref{conject1}." This refutation is achieved by establishing a lower bound, denoted $G_m$, on the number of non-zero terms in a localized diagonal form. Remarkably, $G_m$ reaches cosmologically large values, comparable to the entropy of the observable universe for certain localities $m$. While this theoretically guarantees the conjecture's falsity, the immense scale of $G_m$ motivates us to explore the implications for practically sized systems through a probabilistic approach. We construct a set of random, non-local, sparse, diagonal forms and hypothesize their probability of finding a local representation. Our hypothesize suggests a sharp transition in this probability, linked to the Hamiltonian's sparsity relative to the Bekenstein-Hawking entropy of neutron stars to black holes transition. This observation hints at a potential connection between Hamiltonian sparsity, localizability, critical phenomena warranting further investigation into their interplay in both theoretical and astrophysical contexts.
\end{abstract}

\section{Introduction}
Given a quantum Hamiltonian $H$, and an initial quantum state $|\psi(0)\rangle$, the time evolution for time $t$ of that quantum state is
\begin{align}
|\psi(t)\rangle = e^{-itH}|\psi(0)\rangle
\ .
\end{align}
Computing quantum state time evolution is a fundamental challenge across various scientific disciplines, including chemistry, biology, and materials science. A notable example is the design of efficient catalysts in enzymatic reactions \cite{tapia2010beyond}.

Diagonalizing quantum Hamiltonians is a powerful tool for fast-forwarding quantum state time evolution, where the fast-forwarded time exceeds the time scales of the Hamiltonian's interacting constituents. If we know the diagonalizing unitary $W$ that diagonalizes $H$ by satisfying
\begin{align}
W^\dagger H W =  D 
\end{align}
where $D$ is a diagonal operator (diagonal form) that only contains the eigenvalue information of $H$, then quantum state time evolution simplifies to
\begin{align}
|\psi(t)\rangle
=
(W)(e^{-it D })(W^\dagger)|\psi(0)\rangle
\end{align}
where the time operator exponential only needs to be evaluated over a diagonal operator. Any operator sum decomposition of $D$ into simpler diagonal operators will be easier to evaluate because all diagonal operators commute. Variational Hamiltonian Diagonalization (VHD) \cite{commeau2020variational} was proposed as a quantum algorithm that could solve quantum Hamiltonian diagonalization for quantum state fast forwarding that could be implemented on noisy intermediate-scale quantum devices (NISQ). 

Similar to the Variational Quantum Eigensolver (VQE) \cite{peruzzo2014variational}, our approach requires the Hamiltonian to be expressed as a weighted sum of unitaries
\begin{align}
H = \sum_{a,b} h_{a,b} (i)^{a\cdot b} X^{a} Z^b 
\end{align}
where
\begin{align}
X^k = \bigotimes_{p=0}^{n-1} \bigxxx(
I \delta_{0,k_p} + X \delta_{1,k_p}
\bigxxx)
\label{pauli x products}
\\
Z^k = \bigotimes_{p=0}^{n-1} \bigxxx(
I \delta_{0,k_p} + Z \delta_{1,k_p}
\bigxxx)
\label{pauli z products}
\end{align}
are the Pauli string operators, $a,b,k\in \Z_2^n$, $k_p\in \Z_2$ are the bits  of $k=\sum_{p=0}^{n-1}2^nk_p$, $h_{a,b}$ are the Pauli string weights, $a\cdot b$ is the Hamming inner product, and 
\begin{align}
I = 
\begin{pmatrix}
1 & 0 \\
0 & 1
\end{pmatrix}
&&
X = 
\begin{pmatrix}
0 & 1 \\
1 & 0
\end{pmatrix}
&&
Z = 
\begin{pmatrix}
1 & 0 \\
0 & -1
\end{pmatrix}
\end{align}
are the identity and Pauli matrices. This unitary sum decomposition is used because the individual unitaries in the sum can be implemented on a quantum computer using quantum gate operations. The Pauli string operators form a complete operator basis in a $n$-qubit Hilbert space.

Throughout this paper, we will be focusing on two important definitions: 

(1) The number of non-zero Pauli string weights :
\begin{align}
\text{NNZ}_{\text{PSW}}(H)
=
\text{NNZ}(h)
=
\sum_{a,b\in \Z_2^n} \mathbbm{1}_{ h_{a,b} \neq 0 }
\end{align}
which counts the number of non-zero Pauli string weights. This is also loosely defined as the zero norm of the Pauli string weights.

(2) The Pauli string qubit locality:
\begin{align}
\text{L}_{\text{PSQ}}(H)
=
\max_{a,b\in \Z_2^n}
\bigx\{
|a\lor b|\text{ : } h_{a,b} \neq 0
\bigx\}
\end{align}
which is the largest number of qubits that are aligned to non-identity Pauli matrices of a single Pauli string with non-zero weight, $|\cdot |$ is the Hamming weight, $\lor$ is the logical bitwise OR.

For the algorithm to run in polynomial time with respect to the Hamiltonian system size (number of qubits $n$), three important requirements need to be satisfied: 

(1) $H$ and $D$ must have no more than a $\mathcal{O}(\text{poly}(n))$ number of non-zero weights in their Pauli string sum representation.

(2) $W$ must have no more than $\mathcal{O}(\text{poly}(n))$ quantum gate circuit depth.

(3) No barren plateau issues caused by the ansatz of $W$ and $D$ which can further constrain (1) and (2).

The ansatz of the diagonal form in \cite{commeau2020variational} was selected to be the Pauli string diagonal operator sum 
\begin{align}
D = \sum_{k=0}^{2^n-1} \gamma_k Z^k = W^\dagger H W
\ ,
\end{align}
where $\gamma_k$ are the diagonal form coefficients. There are two unsolved questions about VHD: 

(1) What is the set of all $H$'s having $D$'s that satisfy $\text{NNZ}_{\text{PSW}}(H)=\mathcal{O}(\text{poly}(n))$ and $\text{NNZ}_{\text{PSW}}(D)=\mathcal{O}(\text{poly}(n))$?

(2) For the set of all $H$'s having $D$'s that satisfy $\text{NNZ}_{\text{PSW}}(H)=\mathcal{O}(\text{poly}(n))$ and $\text{NNZ}_{\text{PSW}}(D)=\mathcal{O}(\text{poly}(n))$, what is its subset  satisfying $\text{L}_{\text{PSQ}}(D) = k$ where $k\geq 0$ and $k$ independent of $n$?

It turns out that answering these two questions is both analytically and numerically difficult, because the eigenvalue spectrum $\lambda$ of $H$ is invariant under permutations of its vectored ordering (i.e. $\lambda\rightarrow \pi \lambda \Leftrightarrow \pi D \pi^\dagger \rightarrow D' : \pi \in S_{2^n}$), but their associated diagonal form  coefficients $\gamma$ are not. $S_{d}$  is the set of all permutation matrices in dimension $d$. 
 $\gamma$ and $\lambda$ are treated as vectors of dimension $2^n$ satisfying
\begin{align}
\gamma(\pi)
=
\dfrac{1}{\sqrt{2^n}}
H^{\otimes n} \pi\lambda
\end{align}
where $\gamma(\pi)$ means $\gamma$ is sensitive to the permutation matrix $\pi$, 
\begin{align}
(H^{\otimes n})_{p,q}
=
\dfrac{1}{\sqrt{2^n}}
(-1)^{p\cdot q}
\end{align}
is the $n$-qubit Hadamard unitary matrix satisfying $(H^{\otimes n})^2 = \mathbbm{1}$.
For example, in $n=4$, the following three diagonal forms have the same eigenvalue spectrum:
\begin{align}
D_{\text{sparse local}} 
=
ZIII
+
IZII
+
IIZI
+
IIIZ
\label{sparse local}
\\
D_{\text{sparse non-local}} 
=
ZIII
+
ZZII
+
ZZZI
+
ZZZZ
\label{sparse non-local}
\end{align}
\begin{align}
D_{\text{dense}} 
=
ZIII
+
IZII
+
IIZI
+
IIIZ
\nonumber\\
-\dfrac{1}{4}\bigxxx(
ZIII
+
ZIIZ
+
ZIZI
+
ZIZZ
\nonumber\\
+
ZZII
+
ZZIZ
+
ZZZI
+
ZZZZ
\bigxxx)
\label{dense}
\end{align}
where $IZIZ$ is short hand notation for $I\otimes Z \otimes I \otimes Z$. $\pi$ satisfying $\pi D_{\text{sparse local}}\pi^\dagger = D_{\text{sparse non-local}}$ can be created using a few CNOT gates. However, $\pi$ satisfying $\pi D_{\text{sparse local}}\pi^\dagger = D_{\text{dense}}$ requires several NOT, CNOT, and CCNOT gates in its construction, because only the first two eigenvalues are permuted. In general, the fewer number of permuted eigenvalues, the more non-zero terms are generated in the diagonal form. This is because $\gamma(\pi)$ and $\lambda$ satisfy an entropic uncertainty inequality \cite{maassen1988generalized}
\begin{align}
h(\pi\lambda) + h(\gamma(\pi)) \geq n
\end{align}
where
\begin{align}
h(\lambda)
=
-
\sum_{k=0}^{2^n-1} 
\dfrac{\lambda_k^2}{
\sum_{r=0}^{2^n-1} \lambda_r^2
}
\log_2
\bigxxx(
\dfrac{\lambda_k^2}{
\sum_{s=0}^{2^n-1} \lambda_s^2
}
\bigxxx)
\end{align}
is the Shannon entropy of the normalized squared $\lambda$. Note: $h(\pi\lambda) = h(\lambda)$.
The entropic uncertainty inequality implies that any small changes in $\lambda$ result in big changes in $\gamma$ and vice versa. It also means that $\lambda$ and $\gamma$ cannot both be sparse. This makes finding optimally sparse representations of $D$ extremely difficult, because a brute-force search across all possible permutations of the eigenvalue spectrum scales $(2^n)!$, which for $n=4$ is infeasible for a common laptop (i.e. $(2^4)!\approx 10^{13}$), and $n=5$ is infeasible for the most powerful supercomputer (i.e. $(2^5)!\approx 10^{35}$). It is also surprising that this topic has not been explored in the literature.

We can combine the entropic uncertainty inequality with Rényi entropy inequalities to lower bound
\begin{align}
\text{NNZ}(\gamma(\pi)) \geq 2^{h(\gamma(\pi))} \geq \dfrac{2^n}{2^{h(\lambda)}}
\label{nnz inequality}
\end{align}
This is a powerful inequality for us, because if we can compute the entropy of the normalized squared eigenvalue spectrum of $H$, we can lower bound $\text{NNZ}(\gamma(\pi))$, and this lower bound is indepedent of $\pi$ . In other words, if $h(\lambda)$ scales asymptotically larger than $\mathcal{O}(\log(n))$ then $\text{NNZ}(\gamma(\pi))$ is guaranteed to scale larger than $\mathcal{O}(\text{poly}(n))$ for all $\pi$. However, we cannot use this inequality to verify that $\gamma(\pi)$ is sparse when $\lambda$ is dense.

Although Eq. \ref{nnz inequality} is independent of eigenbasis permutations; it has a potential problem when it comes to locality. Two diagonal forms $D$ and $D'$ can have the same coefficients but different sets of bitstrings, thus they are indistinguishable in Eq. \ref{nnz inequality}. In other words, $D$ and $D'$ can have the same sparisty, but one could be $k$-local and the other is non-local in their bitstring representation (e.g. $\text{L}_{\text{PSQ}}(D) = k$ and $\text{L}_{\text{PSQ}}(D') = n$). For example, Eq. \ref{nnz inequality} cannot distinguish between the diagonal forms in equations $\ref{sparse local}$ and $\ref{sparse non-local}$.

Since different localities can have the same sparisity, a natural question arises; can all sparse non-local diagonal forms be mapped to local diagonal forms using eigenbasis permutations (i.e. they have the same eigenvalue spectrum)? We will prove that this is indeed false.

\section{Summary Outline of Results}
We show that our Conjecture \ref{conject1} (Non-Local to Local Eigenbasis Permutation Map) is false.  We do this by using some simple but involved combinatorics.  The complexity statement in the Conjecture \ref{conject1} is shown to be false, but the coefficient to do so is quite large. 

We introduce the conjecture in parts.  The first Conjecture \ref{conject1} is a direct translation from the physics to math using vectors with bounded support.  An alternate, easier to work with Conjecture \ref{conject2} that uses single elements of the support is also given, and the two conjectures are proved to be the same.

In the next section we cover some ``Tools'' that will be used to create the bounds.  We cover some notation of functions, Fourier transforms, and characters.  A map $\Psi$ is given that is shown to be an injective homomorphism between the bit vectors under XOR and a certain group ring under multiplication.  Finally, a few technical Lemmas are presented with results that the images of $\Psi$ times a certain two power have integral coefficients.

The last section extends a series of increasingly larger bounds, $A_m$, $B_m$, $D_m$, $E_m$ and $G_m$.  This is done mainly by considering sets and counting (combinatorics).  The fact that $\Psi$ is a homomorphism is used near the end of the Lemma \ref{nonzeromult} by inputting a subgroup $S$ of large enough size that we see that $\Psi$ would have to have repeats.  This contradicts that $\Psi$ is an injective homomorphism.  Regarding $S$ also as the set in the conjecture, the refutation follows.

To more concisely state our result: Given the set of all $n$-qubit diagonal forms $D$ and three additional subsets
\begin{align}
R &=& 
\bigx\{
D 
\ &|& \ 
\text{NNZ}_{\text{PSW}}(D)\geq  2^{\lceil \log_2(G_m)\rceil} 
\bigx\}
\\
R' &=&
\bigx\{
D 
\ &|& \ 
\text{L}_{\text{PSQ}}(D)\leq m
\bigx\}
\\
R'' &=&
\bigx\{
\pi D \pi^\dagger 
\ &|& \ 
D \in R' \text{ and } \pi \in S_{2^n}
\bigx\}
\end{align}
where $m \leq n$, $m$ is not a function of $n$, and 
\begin{align}
G_m
=
(2^m + 1) ^ {m^{(3m)} 64^{(m^2-m)}}
\label{Gm}
\end{align}
then $R \setminus R''$ is non-empty, meaning there exists scenarios where sparse non-local diagonal forms cannot be shown to be eigenspectrum equivalent to local diagonal forms.

\section{Discussion}

%%%%%"This gives the main result of this article, namely that the black hole microstate degeneracy, equal to the number of possible orthogonal states in a given energy band is equal to the exponential of the Bekenstein-Hawking entropy (14)" A microstate (W) is a specified configuration of a system's particles and their respective locations and energies. Hawking Entropy of neutron star to black hole Transition = $2^{260}$

\begin{figure}[!h]
\centering
\includegraphics[width=0.9999\linewidth]{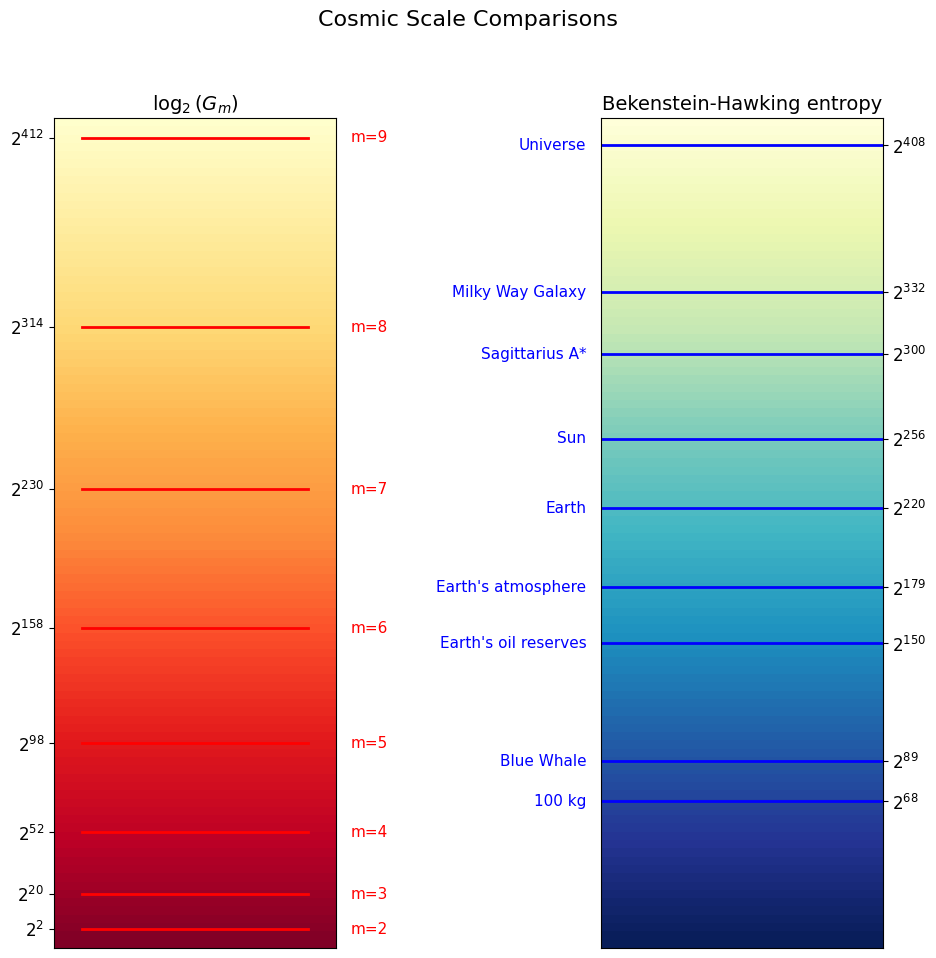}
\captionsetup{width=.9\linewidth}
\caption{
Cosmic Scales and the Limits of Localizability: This figure compares the theoretical bound $G_m$ (Eq. \ref{Gm}), governing the refutation of the Quasiparticle Locality Conjecture \ref{conject1}, to various cosmic scales represented by their Bekenstein-Hawking entropy (in bits) \cite{balasubramanian2024microscopic}. The immense magnitude of $G_m$, reaching universe-scale entropies, suggests that while demonstrably false in theory, its practical implications might only arise for systems as extreme as black holes. Notably, the neutron star to black hole transition occurs between $m = 7$ and $m = 8$, corresponding to typical entropies of stellar-mass black holes. (The bar chart displays the maximum entropy attainable by a given mass when compressed into a black hole).
}
\label{cosmic scales}
\end{figure}

This paper establishes the falsity of Conjecture \ref{conject1}. This means that a sparse, non-local diagonal form of a quantum Hamiltonian cannot always be transformed into a local diagonal form via eigenbasis permutations, using unitary transformations. The theoretical proof relies on a large constant ($G_m$) that dictates the system size $n$ needed for the refutation to hold. While Conjecture \ref{conject1} is demonstrably false for sufficiently large systems, the size of $G_m$ raises practical concerns. For realistically sized systems, encountered in computational simulations and describing physical phenomena, the system size $n$ might not reach the theoretical limit required for Conjecture \ref{conject1} refutation to have tangible implications. This suggests a crucial gap between the theoretical limitations exposed by the Conjecture \ref{conject1}'s falsity and their practical consequences for quantum systems that we can realistically study and simulate.

We propose to investigate a new research direction centered around the hypothesis that black holes may fundamentally be non-quasiparticle systems. Although black holes have been extensively studied within the framework of the Sachdev-Ye-Kitaev (SYK) model \cite{sachdev2023quantum}, the question of whether they can be definitively classified as non-quasiparticle systems remains unresolved.

If black holes are indeed non-quasiparticle in nature, this raises a critical and intriguing question: How does a neutron star, which is traditionally understood as a quasiparticle system, transition into a black hole during gravitational collapse? Specifically, we seek to understand the mechanism by which the quasiparticle description of a neutron star breaks down as it collapses into a black hole, potentially leading to a non-quasiparticle description. We aim to investigate the nature of this hypothetical transition and to uncover the underlying physical mechanisms that might govern it.

In our analysis, we employed the spin-z diagonal form representation $D$ to investigate the system's properties. To connect this representation with the fermionic picture of quasiparticles, we utilized the Jordan-Wigner transformation, a well-established method that maps spin operators to fermionic operators. Specifically, the Jordan-Wigner transformation allows us to map spin-z local diagonal operators to local diagonal fermionic operators. This mapping is achieved by representing each spin-$1/2$ in terms of fermionic creation and annihilation operators, which obey anti-commutation relations. This correspondence allows us to translate the problem from a spin system into a fermionic system, facilitating the analysis within the spin-z representation. By using this approach, we ensure that the key features of the spin system are preserved in the fermionic framework, enabling a more comprehensive understanding of the underlying physics.

Having established the theoretical limit on localizing diagonal forms through the refuted Conjecture \ref{conject1}, we now examine its practical implications by analyzing $D_R(k)$. The sharp transition in probability at $k = 2^w$ links the mathematical properties of Hamiltonian sparsity to the critical particle count relevant to neutron star to black hole transitions, hinting at a profound connection between seemingly disparate fields. As shown in Fig. 1, the scale of $G_m$ reaches values comparable to the Bekenstein-Hawking entropy \cite{balasubramanian2024microscopic} of stellar-mass black holes for $m = 7$ and $m = 8$, placing this transition between these qubit localities. This begs the question: Does the difficulty of finding a localizable representation, as quantified by $G_m$ for practical $m$, become intertwined with critical phenomena like the transformation from neutron star to black hole?

To explore this question, we propose a hypothesis. We introduce a probabilistic approach by considering a random non-local sparse diagonal form, denoted as $D_R(k)$. This diagonal form is constructed as follows:

\begin{align}
D_{R}(k) = 
\sum^{k \ : \ k < n}_{
\substack{
j=1 \ : \ \gamma_j \in \mathbbm{R}
\\
t_1, t_2, \cdots, t_k \in_R \mathbbm{Z}_{2^n}
}
}
\gamma_{j} Z^{t_j}
\end{align}
where $k$ (the sparsity parameter) controls the number of non-zero terms in the summation,  $\gamma_j$  represents a randomly chosen real coefficient for each term, and  $Z^{t_j}$  is a Pauli-Z string operator associated with a randomly selected bitstring  $t_j$. By analyzing  $D_R(k)$, we aim to determine the probability of finding a local representation for a sparse diagonal form chosen randomly from the space of all possible sparse diagonal forms, thus providing insights into the prevalence of localizability in practically relevant quantum systems.

We suggest the hypothesis that the randomly constructed diagonal form $D_R(k)$ reveals a striking result: the probability of finding a local representation through eigenbasis permutations depends critically on the relationship between the sparsity parameter $k$ and the Bekenstein Hawking entropy relevant to neutron star-black hole transitions, denoted as $w$.  

This probability can be expressed as:

\begin{equation}
\mathbb{P} 
\left( 
\bigcup_{\pi \in S_{2^n}}
\bigxxxx( 
\text{L}_{\text{PSQ}}
\bigxxx( 
\pi
D_{R}(k)
\pi^\dagger
\bigxxx)
\leq m
\bigxxxx)
\right)
\approx
\begin{cases} 
1 & \ : \ k < 2^w \\
0 & \ : \ k > 2^w
\end{cases}
\end{equation}

This equation indicates an abrupt transition:  when the sparsity parameter  $k$ is less than  $2^w$, the probability of finding a local form is close to 1 (almost certain). However, when $k$  exceeds $2^w$, this probability plummets to near zero, indicating that finding a local representation becomes highly improbable.

The sharp transition in localizability probability, governed by the relationship between the sparsity parameter $k$ and the critical particle count $w$, carries significant implications for both quantum simulation and our understanding of astrophysical phenomena. Fig. \ref{cosmic scales} suggests that this might be possible, because the critical Bekenstein-Hawking entropy occurs between $m=7$ and $m=8$.

The findings presented in this paper open up compelling avenues for future research, bridging the realms of quantum information science and astrophysics.  Several key directions warrant further investigation:

While our analysis establishes a connection between $k$ and $w$, deriving a more precise mathematical relationship between the sparsity parameter and the critical particle count for neutron star-black hole transitions is crucial. This could involve exploring analytical and numerical techniques to characterize the functional dependence of $k$ on $w$ and investigating potential scaling laws. 

By pursuing these research directions, we can deepen our understanding of the interplay between sparsity, localizability, and critical behavior in quantum systems, potentially uncovering new insights into the nature of quantum complexity and the fundamental workings of the universe.

\section{Conclusion}
Conclusion
This work unveils a fascinating interplay between the abstract mathematical properties of quantum Hamiltonians and the complex physics of extreme astrophysical phenomena. We've demonstrably disproven the Conjecture \ref{conject1}, revealing a fundamental limitation on the localizability of certain sparse diagonal forms via eigenbasis permutations. While the theoretical guarantee of this refutation rests upon the magnitude of a large constant $(G_m)$, this work explores its implications within the more practically relevant domain of finite-sized systems.
By introducing a probabilistic framework and analyzing a randomly constructed sparse diagonal form $(D_R(k))$, we've uncovered an abrupt transition in localizability probability. This transition hinges upon a critical relationship between the Hamiltonian's sparsity and a critical particle count relevant to neutron star-black hole transformations. The emergence of this connection suggests that the feasibility of efficiently simulating quantum systems, particularly those nearing critical points, may be deeply intertwined with the sparsity structure of their underlying Hamiltonian descriptions. Furthermore, this research hints at a profound interplay between the mathematical framework of quantum information and the physical processes governing extreme astrophysical systems.
Future explorations should focus on unraveling the precise nature of the relationship between Hamiltonian sparsity and critical phenomena, particularly in the context of neutron star-black hole transitions. Investigating this connection across various physical systems and developing sophisticated simulation models promise to unlock profound insights into the fundamental principles governing quantum complexity, efficient simulatability, and the dramatic transformations that shape our universe.

\section{Acknowledgements}

The research in this presentation was conducted with the U.S. Department of Homeland Security (DHS) Science and Technology Directorate (S\&T) under contract number 75N98120D00172 and task number 70RST23FR0000015. Any opinions contained herein are those of the author and do not necessarily reflect those of DHS S\&T.

%\printbibliography %Prints bibliography

\appendix

\section{Supplementary Information}

\subsection{Conjecture}
Let $n$ be a positive integer with elementary abelian two group $\V$ corresponding to the group of $n$-bit vectors under XOR.  For $v = (v_1,\cdots,v_n) \in \V$ let $|v| = \sum_{i=1}^n v_i$ define the locality of $v \in \V$.  This is also called the weight, density, or popcnt.

Let $\pi \in Sym(\V)$ be a permutation which we also regard as a $2^n$ by $2^n$ permutation matrix.  Let $H$ be the $2^n$ by $2^n$ Hadamard matrix
\begin{align}
H_{I,J} = (-1)^{I \cdot J}.
\end{align}
and let $e_i$ be a standard basis vector $(\delta_{x,i})_{x \in \V}$, a density one column.

\begin{conjecture}[Non-Local to Local Eigenbasis Permutation Map]
\label{conject1}
For $k \in \Z_{\ge 0}$ there exists an $m_k \in \Z$ and some constant $C_k$ such that if $S \subseteq \V$ has size $|S| \le C_k n^k$ then for any $x = \sum_{J \in S} c_J e_J$ there exists a permutation $\pi$ such that for any $I \in \V$ with $|I| > m_k$ we have $(H \pi H x)_{I}$ = 0.
\end{conjecture}

\begin{conjecture}[Subset Formulation]
\label{conject2}
For $k \in \Z_{\ge 0}$ there exists an $m_k \in \Z$ and some constant $C_k$ such that if $S \subseteq \V$ has size $|S| \le C_k n^k$ then there is a permutation $\pi$ such that for any $J \in S$ and any $I \in \V$ with $|I| > m_k$ we have $(H \pi H)_{I,J}$ = 0.  
\end{conjecture}

Note, we will only use the $k=0$ case, so we write $m = m_k = m_0$.

\begin{lemma}
\label{conjeq}
  Conjectures \ref{conject1} and \ref{conject2} are equivalent.
\end{lemma}
\begin{proof}
First we assume that Conjecture \ref{conject1} is true.  Let $\tau \in \R$ be a transcendental number\footnote{Alternatively, we could use a set of vectors that are approximately linearly independent.}, and let $S = \{J_1,\cdots,J_{|S|}\}$.  Define $x = \sum_{i = 1}^{|S|} \tau^i e_{J_i}$.  Then there is a $\pi$ such that for each $I$, $|I| > m$, we have $(H \pi H x)_I = 0$, and then
\begin{align}  
\left(\sum_{i = 1}^{|S|} (H \pi H) \tau^i\right)_{I,J_i} = \left(\sum_{i = 1}^{|S|} H \pi H e_{J_i} \tau^i\right)_I 
\nonumber \\
= \left(H \pi H x\right)_I = 0.
\end{align}
For any $J \in S$ we have $J = J_i$ for some $i$.   The transcendence of $\tau$ lets us isolate a single term on the left side $(H \pi H)_{I,J}$ = 0.  Thus Conjecture \ref{conject2} is true.

Second we assume that Conjecture \ref{conject2} is true, and we let $x = \sum_{J \in S} c_J e_J$.  So $I \in \V$ with $|I| > m$.  Then we know that 
\begin{align}
(H \pi H x)_I = \sum_{J \in S} c_J (H \pi H)_{I,J} = 0
\end{align}
and we have that Conjecture \ref{conject1} is true.  
\end{proof}

\subsection{Tools}

\subsubsection{Group Ring and Characters}
Let $R$ be a ring and $G$ be an group.  Then the group ring $RG$ is defined using set of formal linear combinations of the form $\sum_g r_g g$ with coefficients $r_g$ in $R$ and group elements $g$ in $G$.  Addition in $RG$ is coefficient-wise $\sum_g r_g g + \sum_g s_g g = \sum_g (r_g + s_g) g$.  Multiplication of elements in $RG$ is given by multiplying the $G$ elements and collecting and adding up the coefficients in $R$
\begin{align}
\left(\sum_{g_1 \in G} r_{g_1} g_1  \right) \left(\sum_{g_2 \in G} s_{g_2} g_2  \right) = \sum_{g_1,g_2 \in G} r_{g_1} s_{g_2} (g_1 g_2)
\nonumber \\
=
\sum_g \left(\sum_{g_1g_2 = g} r_{g_1} s_{g_2}\right) g
\end{align}
which we recognize as a convolution.

The set of functions $\text{Func}[G,R]$ with pointwise addition and convolution form a ring that is isomorphic to the group ring $RG$.  The map between them is given by taking the coefficients as a function $f(g) = r_g$.

We will first assume that $G$ is abelian so that the characters of $G$ are simply the homomorphisms from $G$ to $\mathbb{C}^*$. We will also only consider $R$ to be $\C$ (or its subrings such as $\Z$). We will study specifically $G = \V$ where the characters are all given by $\chi_x(g) = (-1)^{x \cdot g}$ for $x \in \V$.  We extend $\chi_x$ linearly to $RG$  
\begin{align}
\chi_x\left(\sum_g r_g g\right) = \sum_g r_g \chi_x(g) = \sum_g r_g (-1)^{x \cdot g}
\end{align}
\subsubsection{Fourier Transform}

This maps to a Fourier transform of functions $\hat{f}(x) = \sum_g (-1)^{x \cdot g} f(g)$. Define a ring isomorphism
\begin{align}
\iota: Func(\V,\C) \rightarrow Matrices_{n,1}(\C)
\end{align}
that takes a function $f(x)$ to a column vector of values $(f(x))_{x \in \V}$.  Then we see how $H$ computes a Fourier transform
\begin{align}
\hat{f} = \iota^{-1}(H \iota(f))
\end{align}
We have another isomorphism mentioned before
\begin{align}
\kappa: Func(\V,\C) \rightarrow \C \V
\end{align}
that takes a function $f(g)$ to the group ring element $\sum_{g\in \V} f(g) g$. For any function norm $||.||$ we induce norms on group ring elements $||f|| = ||\kappa^{-1}(f)||$.

\begin{lemma}
\label{kappalemma}
\begin{align}
\chi_x(\kappa(f)) = \hat{f}(x)
\end{align}
\begin{proof}
Both sides are equal to $\sum_g f(g) (-1)^{g \cdot x} $.
\end{proof}
\end{lemma}

Let $\pi$ be fixed and let $I,J \in \V$ then
\begin{definition}[Fourier Coefficients $a_I^J$]
\begin{align}
a_I^J := \frac{1}{2^n} \widehat{F_J}(I) = (H \pi H)_{I,J}
\end{align}
\end{definition}
\begin{definition}[Partial Fourier Transform $F_J$]
\label{fdef}
Define the function $F_J$ as $F_J(x) = (-1) ^ {\pi(x) \cdot J}$.
\end{definition}

\begin{definition}[Full Fourier Transform $\overline{F_J}$]
\label{lift}
We define a useful group ring element
\begin{align}
\overline{F_J} := \sum_{g \in \V} a^J_g g \in \C \V
\end{align}
\end{definition}

We can see both input and output characters, $\chi_I$ and $\chi_J$ respectively, occurring as terms here
\begin{align}
\overline{F_J}(I) = \frac{1}{2^n} \sum_{x \in \V} (-1)^{x \cdot I + \pi(x) \cdot J}
\end{align}
so we call $\overline{F_J}$ a ``full'' transform and $F_J$ a ``partial'' transform respectively.

We also have $||\overline{F_J}||_\infty \le 1$ since $|\widehat{F_J}(g)| \le 2^n$.  
To be compatible with Conjecture \ref{conject2} we assume that $a^J_I$ is zero for all $I$ with $|I| > m$.

\begin{lemma}
\label{otherway}
\begin{align}
F_J(x) =  \sum_{g \in \V} a^J_g (-1)^{g \cdot x} = \chi_x(\overline{F_J}).
\end{align}
\end{lemma}
\begin{proof}
  The first equality follows from the inverse Fourier transform.  The second equality uses the definition of $\chi_x$.
\end{proof}

\begin{lemma}
  \label{allchi}
  Let $y \in \C \V$.  Then if  $\chi_x(y) = 0$ for all $x$ then $y = 0$.  Similarly, if $\chi_x(y_1) = \chi_x(y_2)$ for all $x$ then $y_1 = y_2$.
\end{lemma}
\begin{proof}
This is due to the Fourier transform being a change of basis and hence an isomorphism.
\end{proof}

\subsubsection{Defining $\Psi$ Injection}

\begin{lemma}
\label{psichi}
For a $\pi$ the map $\Psi:\V \rightarrow \left(\C \V \right)^*$ that sends $J \mapsto \overline{F_J}$ is an injective homomorphism.  Also $(\chi_x \Psi) (J) = F_J(x)$ and $\Psi_\pi(0) = 1$.
\end{lemma}
\begin{proof}
Using Lemma \ref{otherway} and Definition \ref{fdef} we have for each $x$
\begin{align}
\begin{split}
  \chi_x(\Psi(J_1 + J_2)) &= \chi_x(\overline{F_{J_1 + J_2}}) \\
&= F_{J_1 + J_2}(x) \\
&= (-1)^{\pi(x) \cdot (J_1 + J_2)} \\
&= (-1)^{\pi(x) \cdot {J_1}} (-1)^{\pi(x) \cdot {J_2}} \\
&= F_{J_1}(x) F_{J_2}(x) \\
&= \chi_x(\overline{F_{J_1}}) \chi_x(\overline{F_{J_2}}) \\
&= \chi_x(\Psi(J_1))\chi_x(\Psi(J_2))\\
&= \chi_x(\Psi(J_1)\Psi(J_2)) \\
\end{split}
\end{align}
which using Lemma \ref{allchi} shows it is a homomorphism.

A rewriting of Lemma \ref{otherway} yields   $(\chi_x \Psi) (J) = F_J(x)$.

If $\Psi(J_1) = \Psi(J_2)$ then $\chi_x(\overline{F_{J_1}})$ = $\chi_x(\overline{F_{J_2}})$ for all $x$.  But this means that $F_{J_1}(x)$ = $F_{J_2}(x)$ for all $x$.  Then using Definition \ref{fdef} we have $\pi(x) \cdot J_1$ = $\pi(x) \cdot J_2$ for all $x$.  But this means that $J_1$ = $J_2$ and we have shown $\Psi$ injectivity.  Additionally, an image of $J$ is order two, and thus a unit in $(\C \V)^*$ as stated.

For $\Psi(0) = 1$, we plug $J=0$ into Definition \ref{fdef} to find $F_0(x) = 1$ for all $x$.  This then shows, using Lemma \ref{otherway}, that $\chi_x(\overline{F_0}) = 1$ for all $x$.  But again using Lemma \ref{allchi} this means that $\overline{F_0} = 1$.
\end{proof}

\subsubsection{Partial Sums Over Weight Classes}

\begin{lemma}
\label{induct}
Let $\delta \in \V$ be weight one ($|\delta|=1$).  Let $\&$ be the bitwise AND on $\V$. Define for suitable $d$ 
  \begin{equation}
  \label{Scoef}
  \mathcal{A}_d := \{-1,(-1+d),\cdots,(1-d),1\}.
  \end{equation}
 Then for a given $d$ and $l$ and values $b_g$ the set inclusion
\begin{equation}
  \label{suma}
  C(x) := \sum_{|\alpha| \le l} b_{\alpha} (-1) ^ {\alpha \cdot x} \in \mathcal{A}_d
  \end{equation}
implies the following other sets inclusions
\begin{equation}
  \label{sumb}
  \begin{array}{ccc}
    D^-(x) & := & \left(\sum_{\begin{array}{c} \beta \text{ \& } \delta = 0,\\ |\beta| \le (l-1) \end{array}} b_{\beta + \delta} (-1) ^ {\beta \cdot x} \right) \in \mathcal{A}_{\frac{d}{2}} \\
     & \text{and} & \\
    D^+(x) & := & \left(\sum_{\begin{array}{c} \beta \text{ \& } \delta = 0,\\ |\beta| \le l \end{array}} b_{\beta
    } (-1) ^ {\beta \cdot x} \right) \in \mathcal{A}_{\frac{d}{2}} \\
  \end{array}
  \end{equation}
\end{lemma}
\begin{proof}
  Suppose we have the set inclusion (\ref{suma}) and taking $\alpha = \delta + \beta$ and also $\alpha = \beta$  we write
\begin{align}
C(x) = (-1)^{\delta \cdot x} D^-(x) + D^+(x)
\end{align}
  and also note $D^\pm(\delta + x) = D^\pm(x)$.  Then for $y$ with $y \text{ \& } \delta = 0$,  we have
\begin{align}
2D^\pm(y) = C(y) \pm C(\delta + y) \in \mathcal{A}_d
\end{align}
from which we find the claimed subset inclusions in equations (\ref{sumb}).
\end{proof}

\begin{lemma}
\label{powercoef}
\begin{align}
\overline{F_J} \in \frac{1}{ 2^{ {m - 1}   }}\Z G.
\end{align}
\end{lemma}
\begin{proof}

Let $g$ index the coefficients of $\overline{F_J}$ as in the Definition \ref{lift} and let $l$ and $\delta$ be as in Lemma \ref{induct}. We will proceed by induction using Lemma \ref{induct} on $|g|$ and $l$.

We start with $l=m$ where we find a starting point $C(x) = F_J(x) \in \{-1,1\}$ for Lemma \ref{induct} with $d=2$.  Let $v \le w$ give a partial ordering on $\V$ when $(\forall i) v_i \le w_i$. Consider a sequence $\delta_0 = 0  < \cdots < \delta_m = g$ with $|\delta_i| = i$ and apply Lemma \ref{induct} repeatedly extracting $D^-(x)$ to find 
\begin{equation}
\label{power2coefeq}
\left(\sum_{\begin{array}{c} \beta \text{ \& } \delta_i= 0,\\ |\beta| \le (m-i) \end{array}} a_{\beta + \delta_i} (-1) ^ {\beta \cdot x} \right) \in \mathcal{A}_{\frac{1}{2^{i-1}}}
\end{equation}
The last equation in the sequence is simply $a_g \in \mathcal{A}_{\frac{1}{2^{m-1}}}$.

We next consider the case $|g| < m$.  We again consider a sequence $\delta_0 = 0 < \cdots < \delta_l = g$ this time stopping short at $l < m$ which matches $i = l$ in equation (\ref{power2coefeq}).  The term $a_g$ appears as a coefficient of $1 = (-1) ^ {0 \cdot x}$.  Any other $h > g$ has a term $a_h$, which are the coefficients of $(-1) ^ {(h-g) \cdot x}$.  The $a_h$ are all in $\frac{1}{ 2^{ {m - 1}   }}\Z$ by induction, and so the only other term $a_g$ will also be in $\frac{1}{ 2^{ {m - 1}   }}\Z$.
\end{proof}

\subsection{Bounding}

\begin{definition}
  For $\overline{F_J}  = \sum a^J_g g$ let 
\begin{align}
\supp(\overline{F_J}) := \{g \in G : a^J_g \not = 0\}.
\end{align}
  In other words, $\supp$ is the set of nonzero terms in $\overline{F_J}$.  Let 
\begin{align}
\nodes(\overline{F_J}) := \{i : a^J_{(g_1,\cdots,g_n)} \not = 0, g_i \not = 0\}.
\end{align}
  We call each bit $i$ a node.  In other words, $\nodes$ is the set of bits that are used in $\overline{F_J}$.  For a set $S \subseteq \V$ , we let $\nodes(S) := \cup_{J \in S} \nodes(\overline{F_J})$. 
\end{definition}

\subsubsection{Bound $A_m$}

\begin{lemma}
\label{boundparseval}
  $|\supp(\overline{F_J})|$ is bounded by a function, which we call $A_m$, of $m$.
\end{lemma}
\begin{proof}
  From Definition \ref{lift} we defined that $a^J_g = \frac{1}{2^n} \widehat{F_J}(g)$, so that Parsevals's theorem tells us that
\begin{align}
\sum_g \left(a^J_g\right)^2 = \left|\left|\overline{F_J}\right|\right|_2 =  1.
\end{align}
Lemma \ref{powercoef} bounds nonzero $|a^J_g|$ below by $\frac{1}{ 2^{ {m - 1}   }}$ or
\begin{align}
\left|a^J_g\right|^2 \ge \frac{1}{ 2^{2m - 2}}.
\end{align}
Summing up these two equations bounds $|\supp(\overline{F_J})|$ above
\begin{align}
|\supp(\overline{F_J})| \le 2^{2m - 2}.
\end{align}
and we can use $A_m = 2^{2m - 2}$.
\end{proof}

\subsubsection{Bound $B_m$}

\begin{lemma}
\label{bound_node}
Recall that $a^J_g$ is zero for all $g$ with $|g| > m$.  Then $|\nodes(\overline{F_J})|$ is bounded by a function, which we call $B_m$, of $m$.
\end{lemma}
\begin{proof}
  This is a direct consequence of Lemma \ref{boundparseval}.  The maximum number of nodes that we can have is given by choosing $m$ different nodes for each term $a^J_g$ and $|\supp(\overline{F_J})|$ of these $g$.  So we can use $B_m = m A_m$.
\end{proof}

\subsubsection{Bound $D_m$}

\begin{lemma}
\label{nonzeromult}
  Let $g \in \V$ be density one ($|g| = 1$) a.k.a. a node.  Let $a,b,c$ be in $\C \V$ each with no terms that use $g$.  Suppose $b$ is nonzero and suppose $a + bg$ and $c$ are both images of $\Psi$.  Then $g \in \nodes((a + bg) c)$.
\end{lemma}
\begin{proof}
  The coefficient of $g$ in $(a + bg) c$ is just $bc$.  $c$ is an image of $\Psi$ and hence has $\chi(c) \not = 0$ for all characters $\chi$.  $b$ is nonzero so that there exists a $\chi$ with $\chi(b) \not = 0$ and hence $\chi(bc) \not = 0$.  Thus $bc \not = 0$n a $g$ is a node.
\end{proof}

\begin{figure}[H]
\centering
\includegraphics[width=0.9999\linewidth]{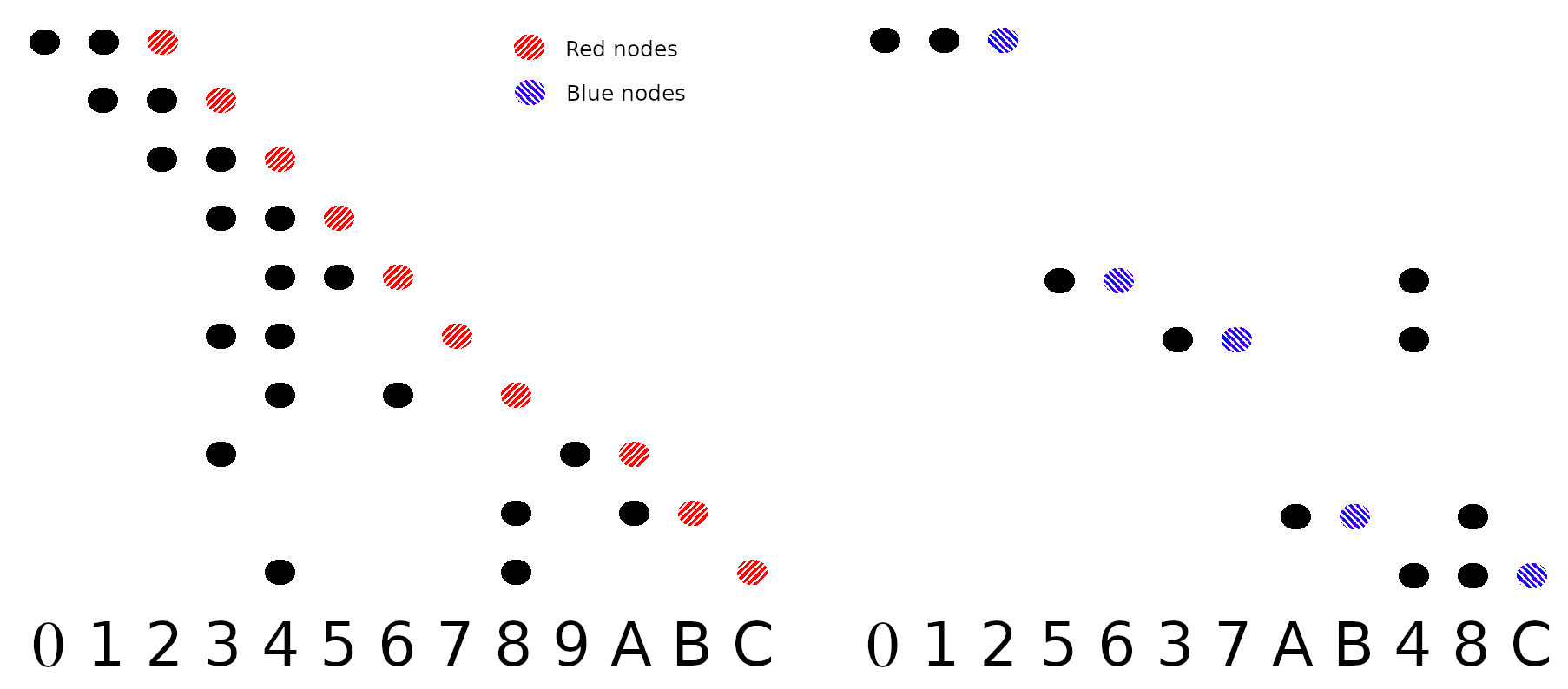}
\captionsetup{width=.9\linewidth}
\caption{Example construction of a $\Psi$-product.  The $i$-th rows of the left hand diagram each correspond to a $\Psi(J_i)$.  The nodes which are colored red are each using a new dimension, so no black nodes are above a red node.  Going from the (red) left half side to the (blue) right half side the columns are permuted so that some red nodes become blue when selected from rows $i_1,\cdots,i_l$.  The blue nodes have the property that no black nodes are above or below them, from which we can construct a product $\Psi(J_{i_1}) \cdots \Psi(J_{i_l})$ that uses all of the blue nodes.}
\label{nodefig} 
\end{figure}

\begin{lemma}
\label{bignodebound}
  Lemma \ref{bound_node} says that $|\nodes(\Psi(J))| \le B_m$ for any single $J$.  We will next consider a set $S \subseteq \V$ of $J$ which is closed under addition.  In other words a set where $S$ is a subgroup of $\V$. We will show that there is an $m$ dependent bound $D_m$ such that $|\nodes(S)| \le D_m$.
\end{lemma}
\begin{proof}
  We pick an $S$ with $|\nodes(S)| > B_m^3$ ($D_m = B_m^3$) and will find a contradiction.  Consider Figure \ref{nodefig} as a reference to follow along.  (In the Figure we are using $B_m = 3$ for illustration.)

Pick $J_i \in S$ such that each $J_i$ is using a new node that $J_1, \cdots, J_{i-1}$ did not use.  We are assuming $|\nodes(S)| > B_m^3$ and every row uses at most $B_m$ new nodes, so we can find at least $1 + B_m^2$ rows in this way.  The rightmost node is colored red for each added row.  It is always the case that this new red node has no black nodes above it in the picture.  We do this for $1 + B_m^2$ rows corresponding to $J_1, \cdots, J_{1+B_m^2}$.

Next we start with the bottom right red node and work our way up and left creating blue nodes on the right from some of the red nodes on the left.  We permute columns so that each next red node column becomes the next column on the right turning the red node blue, filling in from right to left.  At each new blue node, the black nodes on the row that precede the blue one and these columns are permuted and copied in.

For example, in the Figure, the first blue node is ``C'' and the two black nodes to the left on the same row are ``8'' and ``4'' in that order from right to left.  The next blue node is on ``B'' preceded by black node ``A'' to the left on the same row.  Notice that we don't follow with blue ``8'' since we already handled that column, so we put the red ``8'' as a black node above the ``8'' column instead.  We skip the red node in column ``A'' since we already handled column ``A''.  Likewise for column ``8''.  The next blue node is from column ``7''.

We lose some potential red nodes (such as ``8'') when they have occurred in a column from a previous blue row.  In the worst case, for each blue node we sacrifice $B_m - 1$ red nodes in columns above the blue node row.  Since we have at least $1 + B_m^2$ red node rows total, we know that we will generate at least $1 + B_m$ blue nodes.

When we create a blue node, we know that there are no black nodes below the blue node since all of the previous rows have handled all of their columns.  Since they were also a red node, the blue nodes must have no black nodes above them either.  So there are no black nodes in any blue column.

Let $i_1,\cdots,i_l$ index the blue rows and consider the product
\begin{align}
\Psi(J_{i_1}) \cdots \Psi(J_{i_l}) = \Psi(J_{i_1} + \cdots + J_{i_l}) = \Psi(J)
\end{align}
where we are using the fact that $\Psi$ is a homomorphism and we are letting $J$ be the sum.  Each blue node is in $\nodes(\Psi(J))$,  by repeated applications of Lemma \ref{nonzeromult}, and we have at least $1+B_m$ blue nodes.  But this violates the Lemma \ref{bound_node} bound $B_m$ on the number of nodes for the individual $J$, and is a contradiction.
\end{proof}

\subsubsection{Bounds $E_m$ and $G_m$}

\begin{lemma}
\label{bound_g}
There is a bound on the number of group ring elements with locality $m$.
\end{lemma}
\begin{proof}
Let $D_m$ be the bound on the collective number of nodes from Lemma \ref{bignodebound}.  The number of group ring elements $y$ with locality $m$ on $D_m$ nodes can be counted.  There are $1 + \cdots + {D_m \choose m} \le E_m$ possible terms in $y$ where we pick an upper bound $E_m := D_m ^ m$.  Lemma \ref{powercoef} with equation (\ref{Scoef}) says that each term is in $\mathcal{A}_{\frac{1}{2^{m-1}}}$.  $\left|\mathcal{A}_{\frac{1}{2^{m-1}}}\right| = 2^m + 1$, so the number of group ring elements is bounded by $G_m := (2^m + 1) ^ {E_m}$.
\end{proof}

\subsection{Refutation}
\begin{theorem}
  The Conjecture \ref{conject1} is not true.
  \label{Refutation theorem}
\end{theorem}
\begin{proof}
To disprove Conjecture \ref{conject1} using Lemma \ref{conjeq} we find a set $S_m \subseteq \V$ whose size is a function of $m$ and which can not be written at locality $m$:
\begin{align}
|\Psi(S_m)| > m.
\end{align}
For big enough $n$ let $S$ be a subgroup of $\V$ with two power size just exceeding $G_m$.  Then Lemma \ref{bound_g} shows that $\Psi(S)$ must contain repeats.  But $\Psi$ is injective, so this is a contradiction.
\end{proof}

\subsection{Large Constant}

Working backwards we can find the rather large constant $G_m$ used in the proof.  It is
\begin{align}
G_m
=
(2^m + 1) ^ {m^{(3m)} 64^{(m^2-m)}}
\end{align}
We are in fact matching $G_m$ with the size of a subgroup $S$.  Lets suppose that it is size $2^p$ then we have $G_m < |S| = 2^p$ and we may want to take a log by picking a basis of $S$.

\end{document}